\documentclass[aps,prl,reprint,longbibliography]{revtex4-2}
\usepackage{amsmath,amssymb,amsthm}
\usepackage{orcidlink}
\usepackage{graphicx}
\usepackage{thmtools}
\usepackage{thm-restate}
\usepackage{booktabs}
\usepackage{float}
\usepackage{mathtools}
\usepackage{mathrsfs}
\usepackage{braket}
\usepackage{qcircuit}
\usepackage{placeins}
\usepackage{enumitem}
\usepackage{subcaption}
\usepackage[percent]{overpic}
\newcommand{\abs}[1]{\left|#1\right|}
\newcommand{\norm}[1]{\left\|#1\right\|}
\newcommand{\tr}{\mathrm{tr}}
\newcommand{\rk}{\mathrm{rank}}

\newcommand{\diag}{\mathrm{diag}}

\newcommand{\calC}{\mathcal{C}}

\newcommand{\Hm}{\mathscr{H}}

\newcommand{\calP}{\mathcal{P}}
\newcommand{\poly}{\mathrm{poly}}
\newcommand{\LHS}{\mathrm{LHS}}

\newcommand{\ketbrasame}[1]{|#1\rangle\langle#1|}

\theoremstyle{definition}

\newtheorem{proposition}{Proposition}

\newtheorem{definition}{Definition}

\usepackage{hyperref}
\hypersetup{
    hidelinks,
    pdftitle={}, 
    pdfauthor={}, 
    colorlinks=true, 
    linkcolor=blue, 
    citecolor=blue, 
    urlcolor=blue, 
   	final=true, 
}
\usepackage[capitalise,nameinlink]{cleveref}
\crefname{fact}{Fact}{Facts}
\crefname{prog}{Program}{Programs}
\crefformat{prog}{#2Prog.~(#1)#3}

\makeatletter
\def\blfootnote{\gdef\@thefnmark{}\@footnotetext}
\makeatother

\begin{document}

\title{Energy Spectra of Compressed Quantum States}

\author{Daochen Wang\,\orcidlink{0000-0001-5472-1207}\textsuperscript{\hyperref[emailnote]{*}}}
\affiliation{%
University of British Columbia, Vancouver, Canada
}

\begin{abstract}
Quantum algorithms for estimating the ground state energy of a quantum system often operate by preparing a classically accessible quantum state and then applying quantum phase estimation. Whether this approach yields quantum advantage hinges on the state's energy spectrum, that is, the sequence of the state's overlaps with the energy eigenstates of the system Hamiltonian. We show that the energy spectrum of \emph{any} entanglement-compressed quantum state must have large support if most energy eigenstates are highly entangled, an assumption supported by the eigenstate thermalization hypothesis.  Furthermore, we show that if the compressed quantum state minimizes expected energy, then its energy spectrum decays with the inverse-squared energy eigenvalues under a convex relaxation of the compression constraint. This explains the main empirical finding of Silvester, Carleo, and White~\cite{unusual_white_2025} (\textit{Physical Review Letters}, 2025) that the energy spectra of matrix product states do not decay exponentially.
\end{abstract}
\maketitle
\blfootnote{\phantomsection\label{emailnote}\textsuperscript{*}Contact author: \href{mailto:wdaochen@gmail.com}{wdaochen@gmail.com}}

Estimating ground state energies of quantum systems is widely regarded as a promising application of quantum computation \cite{elucidataing_reiher_2017,qchem_survey_cao_2019,qchem_survey_bauer_2020,qchem_survey_mcardle_2020,catalysis_vonburg_2021,qchem_riverlane_2022,evaluating_chan_2023,drug_santagati_2024,spiers_chan_24,fermions_anschuetz_2024,dissipative_lin_2025}. Yet, without further qualifying the system, this problem is believed to be intractable even for quantum computers \cite{yellow_ksv_2002,complexity_watrous_2009}. This raises a fundamental question: \emph{for which} systems could quantum computers offer an advantage over classical computers?

A leading quantum algorithmic paradigm for ground state energy estimation \cite{qpe_abrams_lloyd_1999,qpe_aspuruguzik_2005,quartic_hastings_2020,dequantize_gharibian_2023,improved_glhp_cade_2023,sparse_hamiltonian_chen_2024,probing_choi_izmaylov_2024,state_prep_fomichev_2024,state_prep_berry_2025,lhp_macomplete_jiang_2025,qpe_overlap_lin_izmaylov_2025,van_vleck_simon_2025,glhp_instantiated_schmidhuber_2025,beating_legall_2025} operates in two steps: (i) prepare a classically accessible initial state $\ket{\psi}$, and (ii) apply a standard quantum subroutine (typically, phase estimation \cite{qpe_kitaev_1995}) to produce an energy estimate. To yield quantum advantage, $\ket{\psi}$'s expected energy should be far from the ground state energy \emph{yet} $\ket{\psi}$'s overlap (i.e., squared inner product) with the ground state should also be high. The first condition prevents classical algorithms from succeeding via step (i), while the second condition makes step (ii) quantumly efficient. For both conditions to be met, it is necessary and sufficient for $\ket{\psi}$ to have high overlap with the ground state \emph{yet} nonnegligible overlap with excited states. The sequence of overlaps is known as $\ket{\psi}$'s \emph{energy spectrum} (or distribution) \cite{state_prep_fomichev_2024,unusual_white_2025}.

We focus on states $\ket{\psi}$ with limited entanglement, known as \emph{compressed quantum states} \cite{unusual_white_2025}. These states constitute a substantial family of classically accessible states that includes matrix product states (MPSs) and tensor network states \cite{mps_schollwock_2011,tns_eisert_2013,tns_orus_2019,mps_review_cirac_2021}. We show that $\ket{\psi}$'s energy spectrum must have large support if most energy eigenstates are highly entangled, an assumption supported by the eigenstate thermalization hypothesis. Furthermore, we show that if $\ket{\psi}$ minimizes expected energy then its energy spectrum decays with the inverse-squared energy eigenvalues under a convex relaxation of the compression constraint. This explains the main empirical finding of Ref.~\cite{unusual_white_2025}  that the energy spectra of MPSs do not decay exponentially. We end by discussing how the result relates to the initial question on quantum advantage.

\paragraph{Energy spectra via entanglement.}

Key to our analysis is a robust notion of matrix rank known as stable rank, as introduced by Rudelson and Vershynin~\cite{stable_rank_rudelson_2007,stable_rank_ipsen_2025}.

\begin{definition}[Stable rank]
Let $L,R\in \mathbb{N}$. The \emph{stable rank} of $0\neq A\in \mathbb{C}^{L\times R}$, denoted $\chi(A)$, is defined to be 
\begin{equation}
    \chi(A) \coloneqq \norm{A}_F^2 \, \big/ \, \norm{A}^2,
\end{equation} 
where $\norm{A}_F$ denotes the Frobenius norm of $A$ and $\norm{A}$ denotes the spectral norm of $A$. The stable rank of every matrix with all entries equal to $0$ is defined to be $0$.
\end{definition}
Note that the stable rank is at most the rank because $\norm{A}_F^2 = \sum_{i=1}^{\rk(A)} \sigma_i^2$ and $\norm{A}^2 = \max_{i=1}^{\rk(A)} \sigma_i^2$, where the $\sigma_i$s are the singular values of $A$.

We adapt stable rank to measure the entanglement of quantum states. We focus on qubit states for simplicity; generalizations are straightforward. For $n\in \mathbb{N}$, we write $[n] \coloneqq \{1,\dots,n\}$. We say $A,B\subset [n]$ is a bipartition if $A\cup B=[n]$, $A\cap B = \emptyset$, and $A,B\neq [n]$.

\begin{definition}[Stable Schmidt rank]\label{def:stable_bond_dim}
    Let $n\geq 2$ be an integer. Let $A,B\subset [n]$ be a bipartition. Let $\ket{\psi}$ be an $n$-qubit state. The \emph{stable Schmidt rank of $\ket{\psi}$ with respect to $A,B$}, denoted $\chi_{A,B}(\ket{\psi})$, is defined to be the stable rank of the unique matrix $\Gamma \in \mathbb{C}^{\{0,1\}^A \times \{0,1\}^B}$ such that 
    \begin{equation}
        \ket{\psi} = \sum_{x \in \{0,1\}^A, \, y\in \{0,1\}^B} \Gamma_{x,y}\ket{x}_A\ket{y}_B.
    \end{equation}
\end{definition}
    We will often omit the subscript $A,B$ from $\chi_{A,B}$ when $A,B$ are clear from context or are immaterial.
    
    Since stable rank is at most the rank, $\chi(\ket{\psi})$ is at most $\ket{\psi}$'s Schmidt rank. Moreover, $\chi(\ket{\psi})$ can be viewed as an alias for the entanglement \emph{min-entropy} of $\ket{\psi}$: if we write $\rho_A$ and $\rho_B$ for the reduced density matrices of $\ket{\psi}$ on $A$ and $B$, and $S_{\min}(\rho) \coloneqq \log_2(1/\norm{\rho})$ for the min-entropy of a density matrix $\rho$, then
 \begin{align}\label{eq:srank_minentropy}
     \chi_{A,B}(\ket{\psi}) =&\frac{1}{\norm{\rho_A}} = \frac{1}{\norm{\rho_B}} =  2^{S_{\min}(\rho_A)} = 2^{S_{\min}(\rho_B)},
 \end{align}
 which follows from $\rho_A = \Gamma\Gamma^\dagger$ and $\rho_B = (\Gamma^\dagger\Gamma)^{\intercal}$, and $\norm{\Gamma}_F^2 = \tr[\Gamma\Gamma^\dagger] = \tr[(\Gamma^\dagger\Gamma)^{\intercal}] =\tr[\rho_A] =\tr[\rho_B] = 1$ and $\norm{\Gamma}^2 = \norm{\Gamma\Gamma^{\dagger}} = \norm{(\Gamma^\dagger\Gamma)^{\intercal}} = \norm{\rho_A} = \norm{\rho_B}$.
 
 Since $S_{\min}(\rho) \in [0,\log_2(\dim(\rho))]$, we have $\chi_{A,B}(\ket{\psi}) \in [1, 2^{n/2}]$ for all $n$-qubit $\ket{\psi}$ and all bipartitions $A,B\subset [n]$.
 
 For all $\alpha\in [0,\infty]$, we have $S_{\alpha}(\rho) \geq S_{\min}(\rho) $, where $S_{\alpha}(\rho) \coloneqq (1-\alpha)^{-1}\log_2(\tr[\rho^\alpha])$  is the $\alpha$-R\'enyi entropy of $\rho$. Moreover, if $\alpha>1$, then $S_{\min}(\rho) \geq (1-1/\alpha)S_{\alpha}(\rho)$ \footnote{If $0\leq \alpha \leq 1$, no constant $c>0$ exists such that $S_{\min}(\rho) \geq c S_{\alpha}(\rho)$ for all $\rho$. For example, take $\rho \in \mathbb{C}^{d\times d}$ to be  $\diag(2^{-1}, (2d-2)^{-1},\dots, (2d-2)^{-1})$, in which case $S_{\min}(\rho) = 1$ but $S_\alpha(\rho) \geq S_1(\rho) = 2^{-1} + 2^{-1}\log_2(2d-2)$.}. This means that all results about R\'enyi entropies, see, e.g., Refs.~\cite{entropy_page_93,mps_verstraete_cirac_2006,area_law_hastings_2007,renyi_mps_schuch_2008,entropy_cft_calabrese_cardy_2009,renyi_afhm_song_2011,eigenstates_keating_linden_wells_2015,eigenstates_rigol_2017,arealaw_vidick_2017,eth_garrison_18,renyi2_brydges_zoller_2019,eigenstate_miao_qiang_2021,eigenstate_miao_qiang_2022,area_anshu_2022,area_renyi_anshu_harrow_soleimanifar_2022,eat_metger_2024,entanglement_childs_2024,area_generalized_kim_kuwahara_2025}, immediately carry over to $\chi(\cdot)$.
 
 A key observation about $\chi$ is that any decomposition of a state with small $\chi$ into a sum of pairwise orthogonal states each with large $\chi$ must contain many terms. In fact, we can give a lower bound involving the harmonic mean: for $v\in \mathbb{R}^k_{>0}$, $\Hm(v) \coloneqq k/(\sum_{i=1}^k 1/v_i)$.

\begin{proposition}\label{prop:coefs_lower_bound}
Let $k,n\in \mathbb{N}$. Let $m, M_1,\dots, M_k>0$. Suppose $\ket{\psi}$ and $\ket{\psi_1},\dots,\ket{\psi_k}$ are $n$-qubit states with
\begin{enumerate}
    \item\label{item:state_cond1} $\ket{\psi} = \sum_{i=1}^k \alpha_i \ket{\psi_i}$ for some $(\alpha_1,\dots,\alpha_k)\in \mathbb{C}^k$,
    \item\label{item:state_cond2} $\braket{\psi_i|\psi_j} = 0$ for all $i\neq j$,
    \item\label{item:state_cond3} $\chi_{A,B}(\ket{\psi}) \leq m$ for some bipartition $A,B\subset [n]$,
    \item\label{item:state_cond4} $\chi_{A,B}(\ket{\psi_i})\geq M_i$ for all $i \in [k]$.
\end{enumerate}

    Then,
    \begin{equation}\label{eq:coefs_lower_bound}
        \sum_{i=1}^k\frac{\abs{\alpha_i}}{\sqrt{M_i}} \, \geq \, \frac{1}{\sqrt{m}}.
    \end{equation}
    In particular,
    \begin{equation}\label{eq:refined_no_terms_lower_bound}
        k\geq \frac{\Hm(M_1,\dots,M_k)}{m}.
    \end{equation} 
\end{proposition}

\begin{proof}
Let $A,B\subset [n]$ be such that \hyperref[item:state_cond3]{conditions 3} and \hyperref[item:state_cond4]{4} are satisfied. Write $\chi$ for $\chi_{A,B}$. Write $\Gamma$ and $\Gamma_i$ for the matrices corresponding to $\ket{\psi}$ and $\ket{\psi_i}$ as defined in \cref{def:stable_bond_dim}. By definition, we have $\chi(\ket{\psi}) = \chi(\Gamma)$, $\chi(\ket{\psi_i}) = \chi(\Gamma_i)$, and 
\begin{align}
    \ket{\psi} =&~\sum_{x\in \{0,1\}^A, \, y\in \{0,1\}^B} \Gamma_{x,y}\ket{x,y},
    \\
    \ket{\psi_i} =&~\sum_{x\in \{0,1\}^A, \, y\in \{0,1\}^B} (\Gamma_i)_{x,y}\ket{x,y}.
\end{align}

Since $\norm{\Gamma}_F^2 = \norm{\Gamma_i}_F^2 = 1$, we have $1/\chi(\ket{\psi}) = \norm{\Gamma}^2$ and  $1/\chi(\ket{\psi_i}) = \norm{\Gamma_i}^2$ for all $i\in [k]$.
    
    Therefore, \hyperref[item:state_cond3]{condition 3} gives
    \begin{equation}\label{eq:combine_1}
        \frac{1}{m} \leq \frac{1}{\chi(\ket{\psi})} = \norm{\Gamma}^2;
    \end{equation}
    while \hyperref[item:state_cond1]{conditions 1} and \hyperref[item:state_cond4]{4},  and the triangle inequality, give
    \begin{align}
        \norm{\Gamma}^2=&~\Bigl\|\sum_{i=1}^k \alpha_i\Gamma_i\Bigr\|^2\leq \Bigl(\sum_{i=1}^k \abs{\alpha_i} \norm{\Gamma_i}\Bigr)^2\notag
        \\
        =&~\Bigl(\sum_{i=1}^k \frac{\abs{\alpha_i}}{\sqrt{\chi(\ket{\psi_i})}}\Bigr)^2 \leq \Bigl(\sum_{i=1}^k \frac{\abs{\alpha_i}}{\sqrt{M_i}}\Bigr)^2.\label{eq:combine_2}
    \end{align}
    Chaining \cref{eq:combine_1,eq:combine_2} gives the main inequality of the proposition, \cref{eq:coefs_lower_bound}.

    For the ``in particular'' part, observe that  \hyperref[item:state_cond1]{conditions 1} and \hyperref[item:state_cond2]{2} imply that $\sum_{i=1}^k \abs{\alpha_i}^2 = 1$. Therefore, \cref{eq:refined_no_terms_lower_bound} follows from applying the Cauchy-Schwarz inequality on \cref{eq:coefs_lower_bound} to obtain $\sum_{i=1}^k 1/M_i \geq 1/m$, and then rewriting it using the definition of the harmonic mean.
\end{proof}

\cref{prop:coefs_lower_bound} can be seen as a standalone result on the ``entanglement of superposition''~\cite{entanglement_superposition_linden_2006,entanglement_superposition_gour_2007,entanglement_superposition_gour_roy_2008}, except it uses a different entanglement measure, stable rank. Changing the measure is important: Ref.~\cite[Example 4]{entanglement_superposition_gour_2007} implies that natural analogues of \cref{prop:coefs_lower_bound} with stable rank replaced by (von Neumann) entropy are ``very'' false. By inspection, we see that analogues of \cref{prop:coefs_lower_bound} with stable rank replaced by rank are also ``very'' false \footnote{Suppose $n$ is even. Let $m\coloneqq n/2$ and $L\coloneqq 2^m$. Then consider $\ket{\psi_1} = \sum_{x\in \{0,1\}^m} A_{xx}\ket{x}\ket{x}$ and $\ket{\psi_2} = \sum_{x\in \{0,1\}^m} B_{xx}\ket{x}\ket{x}$, where $A,B\in \mathbb{C}^{L\times L}$ are defined by $A \coloneqq \diag(1/\sqrt{2(L-1)},\dots,1/\sqrt{2(L-1)},1/\sqrt{2})$ and $B \coloneqq \diag(-1/\sqrt{2(L-1)},\dots,-1/\sqrt{2(L-1)},1/\sqrt{2})$. Clearly, $\braket{\psi_1|\psi_2} =0$ and $\rk(A) = \rk(B) = L$, yet $\rk(A+B) = 1$, which can be much smaller than $L$.}.

We proceed to imbue \cref{eq:refined_no_terms_lower_bound} of \cref{prop:coefs_lower_bound} with physical meaning. Suppose $\ket{\psi}$ is an $n$-qubit MPS with bond dimension $D$ that is polynomial in $n$, written $D \leq \poly(n)$. Then, across any bipartition aligned with the MPS chain, we have
\begin{equation}\label{eq:low_m}
  \chi(\ket{\psi}) \leq D \leq \poly(n),  
\end{equation}
because stable rank is at most the rank and the bond dimension is the same as the rank in this setting. 

Now, suppose $\ket{\psi_1},\dots, \ket{\psi_{2^n}}$ are energy eigenstates of a generic $n$-qubit lattice Hamiltonian $H$ satisfying the eigenstate thermalization hypothesis (ETH) \cite{eth_deutsch_91,eth_srednicki_94,eth_deutsch_review_18}. Then, we expect most $\ket{\psi_i}$s to have volume-law entanglement, meaning the min-entropy of their reduced density matrices $\rho_i$ on a sub-lattice of $V \leq n/2$ qubits should be proportional to $V$. More precisely, Ref.~\cite[Section V]{eth_garrison_18} shows that, to leading order, the ETH gives
\begin{equation}\label{eq:Smin_volume_law}
    S_{\min}(\rho_i) = V \cdot s_{\min}(\beta_i) \quad \text{for all $i\in [2^n]$},
\end{equation}
where $s_{\min}(\beta_i)$ is the thermal min-entropy density of the Gibbs state $\propto \exp(-\beta_i H)$ that has expected energy equal to $E_i \coloneqq \bra{\psi_i}H\ket{\psi_i}$. 

Choosing $V$ to be proportional to $n$, we expect $s_{\min}(\beta_i)$ to be roughly constant for most $i$s and so we expect volume-law behavior, i.e., $S_{\min}(\rho_i) \propto n$, for most $i$s. By \cref{eq:srank_minentropy}, this means $\chi(\ket{\psi_i}) \geq 2^{\Omega(n)}$ for most $i$s, and so
\begin{equation}\label{eq:high_HmM}
    \Hm(\chi(\ket{\psi_i}),\dots, \chi(\ket{\psi_{2^n}})) \geq 2^{\Omega(n)}.
\end{equation}

In \cref{app:volume_law}, we formally derive \cref{eq:high_HmM} under the common assumption that the density of states, i.e., the sequence of $E_i$s, follows a Gaussian distribution~\cite{eigenstate_miao_qiang_2021}.

Thus, substituting \cref{eq:low_m,eq:high_HmM} into \cref{eq:refined_no_terms_lower_bound} of \cref{prop:coefs_lower_bound} shows that $\ket{\psi}$ must have nonzero weight on at least $2^{\Omega(n)}/\poly(n)$ energy eigenstates, which is a large number of eigenstates. In other words, the energy spectrum of $\ket{\psi}$ must have large support.

In general, \cref{prop:coefs_lower_bound} places $(n-1)$ constraints on $\ket{\psi}$, one for each cut aligned with the MPS chain. (A cut corresponds to a choice of bipartition.) In practice, we expect the constraint from the \emph{single} cut at the \emph{middle} to subsume or dominate the others because $S_{\min}(\rho_i)$, and hence $\chi(\ket{\psi_i})$, is expected to be maximized at $V = n/2$ --- see \cref{eq:Smin_volume_law} \footnote{The mathematical argument is that the inequality $\sum_{i=1}^k \abs{\alpha_i}/\sqrt{M_i} \, \geq \, 1/\sqrt{m}$ \emph{implies} the inequality $\sum_{i=1}^k \abs{\alpha_i}/\sqrt{M_i'} \, \geq \, 1/\sqrt{m}$ if $M_i\geq M_i'$ for all $i$.}. Henceforth, we will use \cref{prop:coefs_lower_bound} only with respect to the middle cut.

As discussed above, \cref{prop:coefs_lower_bound} places a constraint on the energy spectrum of \emph{every} MPS $\ket{\psi}$ of low bond dimension. However, we may hope to tighten that constraint if $\ket{\psi}$ \emph{additionally} has low expected energy, for example, if $\ket{\psi}$ is the output of DMRG \cite{dmrg_white_92}. The next proposition shows that in the extremal case when $\ket{\psi}$'s expected energy is minimal subject to the constraints in \cref{prop:coefs_lower_bound}, then we can completely characterize its energy spectrum.

\begin{restatable}{proposition}{energymin}\label{prop:energy_min}
    Let $k\in \mathbb{N}$. Let $m,M_1,\dots,M_k>0$. Let $E_1,\dots,E_k \in \mathbb{R}$ be such that $E_1\leq E_2\leq \dots\leq E_k$. Let $l\in [k]$ be the size of the set $\{i\in [k] \colon E_i = E_1\}$.
    
    Suppose
    \begin{alignat}{2}
      &l~<~&&\Hm(M_1,\dots,M_l)/m \label{eq:non_triviality_l},
      \\
      &k~>~&&\Hm(M_1,\dots,M_k)/m. \label{eq:non_triviality_k}
    \end{alignat}
    
    Then, the solution to the program
    \begin{alignat}{2}
        &\mathrm{minimize} \quad &&\sum_{i=1}^k \abs{\alpha_i}^2 E_i\label[prog]{prog:primal}
        \\
        &\mathrm{s.t.} \quad && \alpha \in \mathbb{C}^k, \, \norm{\alpha}_2 = 1\notag
        \\
        &\quad && \sum_{i=1}^k\frac{\abs{\alpha_i}}{\sqrt{M_i}} \, \geq \, \frac{1}{\sqrt{m}}\notag
    \end{alignat}
    is 
    \begin{equation}\label{eq:dual}
        \max_{\nu < E_1} \quad  \nu +\Bigl(m \sum_{i=1}^k \frac{1}{M_i(E_i - \nu)}\Bigr)^{-1};
    \end{equation}
    and the maximum in \cref{eq:dual} is uniquely attained at the unique solution $\nu^* < E_1$ to
    \begin{equation}\label{eq:dual_optimum_condition}
    \frac{1}{m} \sum_{i=1}^k \frac{1}{M_i(E_i-\nu^*)^2} = \Big(\sum_{i=1}^k\frac{1}{M_i(E_i-\nu^*)}\Bigr)^2.
    \end{equation}

    Moreover, every minimizer $\alpha\in \mathbb{C}^k$ of \cref{prog:primal} obeys:
    \begin{equation}\label{eq:inv_square}
        \abs{\alpha_i}^2 \propto \frac{1}{M_i(E_i-\nu^*)^2} \quad \text{for all $i\in [k]$.}
    \end{equation}
\end{restatable}

\begin{proof}
    See \cref{app:proof_energymin}, where we prove the result using convex duality~\cite[Chapter 5]{cvx_boyd_vandenberghe_2004}.
\end{proof}

Note that \cref{eq:non_triviality_l,eq:non_triviality_k} ensure nontriviality:
\begin{enumerate}[leftmargin=*]
    \item If $l\geq \Hm(M_1,\dots,M_l)/m$, then the minimum is simply $E_1$, attained at any $\alpha$ with $\abs{\alpha_i}^2 = M_i^{-1}/\textstyle\sum_{i=1}^l M_i^{-1}$ for all $i\in [l]$ and $\abs{\alpha_i} = 0$ for all $i>l$.
    \item If $k < \Hm(M_1,\dots,M_k)/m$, then there is no feasible $\alpha$. If $k = \Hm(M_1,\dots,M_k)/m$, then the minimum is $m \sum_{i=1}^k E_i/M_i$, attained at any $\alpha$ with $\abs{\alpha_i}^2=m/M_i$ for all $i\in[k]$. 
    (Both statements follow from the Cauchy-Schwarz inequality.)
\end{enumerate}

How should we interpret \cref{prop:energy_min}? The constraint it employs comes from  \cref{eq:coefs_lower_bound} of \cref{prop:coefs_lower_bound}. This constraint applies to every low-bond-dimension MPS $\ket{\psi}$ but it also applies to a \emph{broader} class $\calC$  of ansatz states. Formally, this is because \cref{eq:coefs_lower_bound} is \emph{implied} by the low-bond-dimension condition on $\ket{\psi}$ but is not equivalent to it. (The implication is by \cref{eq:low_m} and \cref{prop:coefs_lower_bound}.) Therefore, the energy spectrum characterized by \cref{eq:inv_square} of \cref{prop:energy_min} yields an expected energy at most that of the actual spectrum of minimal-energy MPSs. This means we expect the tail of the former to place a \emph{floor} (or lower bound) on the tail of the latter. Now, \cref{eq:inv_square} characterizes $p_i\coloneqq |\alpha_i|^2$ by an inverse-square law in terms of $E_i$. Therefore, if the $M_i$s are reasonably well-behaved, this would result in a flat tail when the $p_i$s are plotted on a logarithmic scale against the $E_i$s. This would then explain the main empirical finding of Ref.~\cite{unusual_white_2025} that the actual tail does not decay exponentially. (More recently, Ref.~\cite{chen_temperature_2025} also observes flat-looking tails, but the authors speculate that those tails are in fact slightly sloped.)

While \cref{prop:energy_min} predicts a flat tail, that  tail may sit at lower \emph{height} than the actual tail if $\calC$ fails to capture other pertinent properties of MPSs. One such property is that the overlap between an MPS $\ket{\psi}$ of bond dimension at most $D$ and the ground state $\ket{\psi_1}$ is bounded by the sum of the top $D$ eigenvalues of its reduced density matrix $\rho_1$, that is,
\begin{equation}\label{eq:gs_overlap_bound}
    \abs{\alpha_1}^2 \coloneqq \abs{\braket{\psi|\psi_1}}^2 \leq \textstyle\sum_{i=1}^D \lambda_i(\rho_1) \eqqcolon \Lambda(D),
\end{equation}
where $\lambda_i(\cdot)$ denotes the $i$th largest eigenvalue. \cref{eq:gs_overlap_bound} follows from the Eckart-Young theorem~\cite{lowrK_eckart_young_1936} --- see, e.g., Ref.~\cite[Lemma 10 of arXiv version]{area_renyi_anshu_harrow_soleimanifar_2022}.

To leverage \cref{eq:gs_overlap_bound}, we can set the value of $m$ in \cref{prop:energy_min} to be such that the optimizer $\alpha$ satisfies $\abs{\alpha_1}^2 = \Lambda(D)$. This can be seen as a theoretically-grounded tightening of $\calC$ around MPSs that incorporates additional knowledge about the system, namely, $\Lambda(D)$. It may be reasonable to assume some knowledge of $\Lambda(D)$ because the sequence of $\lambda_i(\rho_1)$s is known as the ground-state \emph{entanglement spectrum}, which has been extensively studied since its introduction by Li and Haldane~\cite{entanglement_li_haldane_2008}.

\paragraph{Case study: 2D AFHM.}

We use \cref{prop:energy_min} to predict spectra and compare them with that from DMRG~\cite{dmrg_white_92}
for the main system considered in Ref.~\cite{unusual_white_2025}: the spin-$1/2$ antiferromagnetic Heisenberg model (AFHM) on an $\ell\times \ell$ grid with periodic boundary conditions (PBC). We present the predictions that result from both using and not using \cref{eq:gs_overlap_bound}, where $\Lambda(D)$ is computed numerically. Following Ref.~\cite{unusual_white_2025}, we focus on the largest $\ell$ that is computationally feasible, $\ell=4$.

The system Hamiltonian on $n \coloneqq \ell^2 = 16$ spins is
\begin{equation}
    H = \frac{1}{4}\sum_{\langle i, j \rangle} X_i X_{j} + Y_i Y_{j} + Z_i Z_{j},
\end{equation}
where the sum ranges over distinct sets $\{i,j\}\subseteq [n]$ such that $i$ and $j$ are neighbors on the grid with PBC, and 
$X \coloneqq \begin{psmallmatrix}
    0 & 1
    \\
    1 &0
\end{psmallmatrix}$,  $Y\coloneqq \begin{psmallmatrix}
    0 & -i
    \\
    i &0
\end{psmallmatrix}$, and $Z\coloneqq \begin{psmallmatrix}
    1 & 0
    \\
    0 & -1
\end{psmallmatrix}$.

We apply \cref{prop:energy_min} with the following parameters:
\begin{enumerate}[leftmargin=*]
    \item When not using \cref{eq:gs_overlap_bound}, $m$ is set to $\chi(\ket{\psi(D)})$, where $\ket{\psi(D)}$ is the MPS approximate ground state of $H$ computed by DMRG with bond dimension $D$ via the TeNPy library \cite{tenpy2024}. When using \cref{eq:gs_overlap_bound}, $m$ is set such that $\abs{\alpha_1}^2 = \Lambda(D)$. The resulting values of $m$ are given below to $4$~significant figures (s.f.).

    \begin{table}[ht]
        \begin{tabular*}{0.925\linewidth}{@{\extracolsep{\fill}} c c c @{}}
            \toprule
            & \multicolumn{2}{c}{$m$ value} \\
            \cmidrule(lr){2-3}
            $D$ & \cref{eq:gs_overlap_bound} used & \cref{eq:gs_overlap_bound} unused \\
            \midrule
            $50$  & $0.02885$ & $1.778$ \\
            $100$ & $0.1074$  & $1.843$ \\
            $150$ & $0.3254$  & $1.863$ \\
            \bottomrule
        \end{tabular*}
    \end{table}
    
    \item For all $i \in [2^n]$, let $\ket{\psi_i}$ be the $i$th energy eigenstate of $H$; set $E_i$ to be its energy and set $M_i \coloneqq \chi(\ket{\psi_i})$. In \cref{fig:2d_comparison_entropy}, we plot $\log_2(M_i) = S_{\min}(\tr_A[\ketbrasame{\psi_i}])$ and the entanglement entropy of $\ket{\psi_i}$, i.e., $S_1(\tr_A[\ketbrasame{\psi_i}])$. Their behaviors are similar, agreeing with ETH that predicts bulk eigenstates  should have larger entanglement (volume law) than edge eigenstates (area law) \cite{eth_garrison_18,eigenstate_miao_qiang_2021,eigenstate_miao_qiang_2022}; also see proof of \cref{eq:high_HmM} in \cref{app:volume_law}.
\end{enumerate} 

\begin{figure}
    \includegraphics[width=0.875\linewidth]{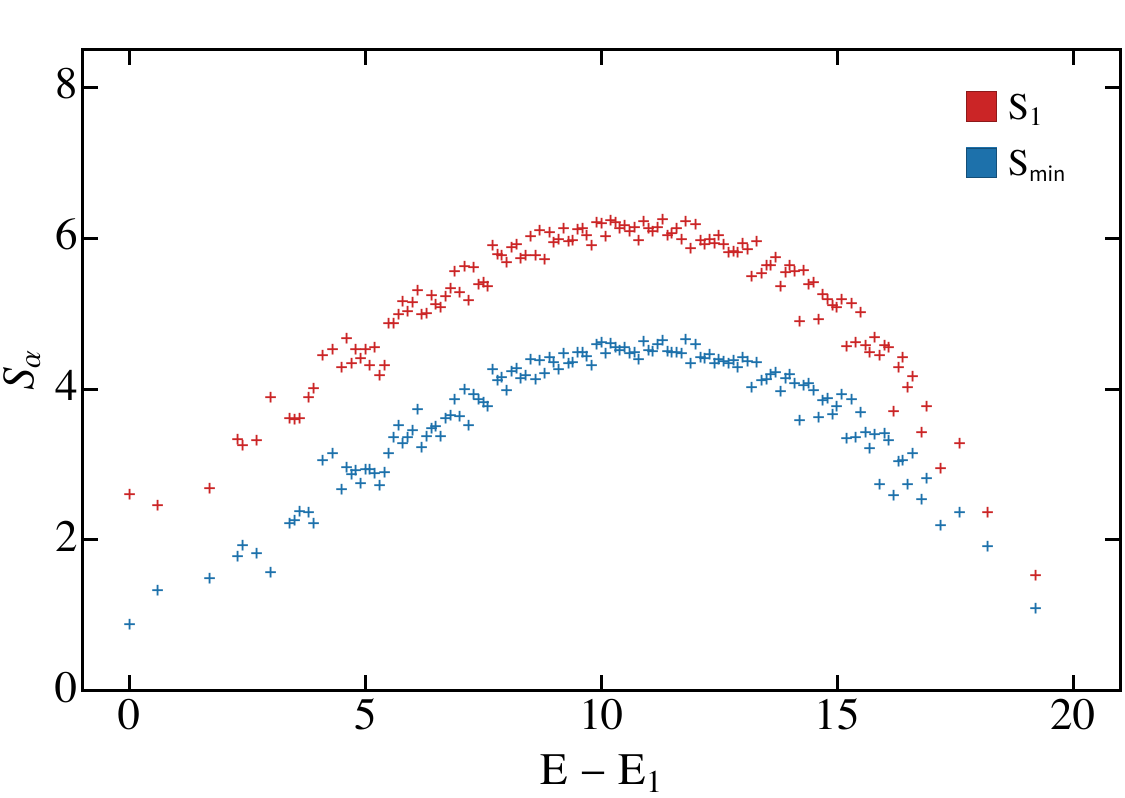}
    \captionsetup{justification=raggedright,singlelinecheck=false}
    \caption{%
    The entanglement min-entropy ($S_{\min}$) and entropy ($S_1$) of eigenstates of the 2D AFHM on $n=16$ spins. Each data point represents the mean of all $S_{\alpha}$ values corresponding to $(E-E_1)$ values within a bin of the form $[0.1j-0.05,0.1j+0.05)$, where $j\in \mathbb{Z}$. All $S_{\alpha}$ values are theoretically at most $\log_2(2^{n/2}) = 8$.
    }
    \label{fig:2d_comparison_entropy}
\end{figure}

In \cref{fig:flat_spectrum}, we plot the predicted and actual energy spectra of the DMRG ground state $\ket{\psi(D)}$, computed with maximum bond dimension $D\in \{50, 100, 150\}$.

\begin{figure}[H]
  \vspace{7pt}
  \begin{subfigure}[b]{0.475\textwidth}
    \begin{overpic}[width=\linewidth, trim = 0 60 0 0, clip]{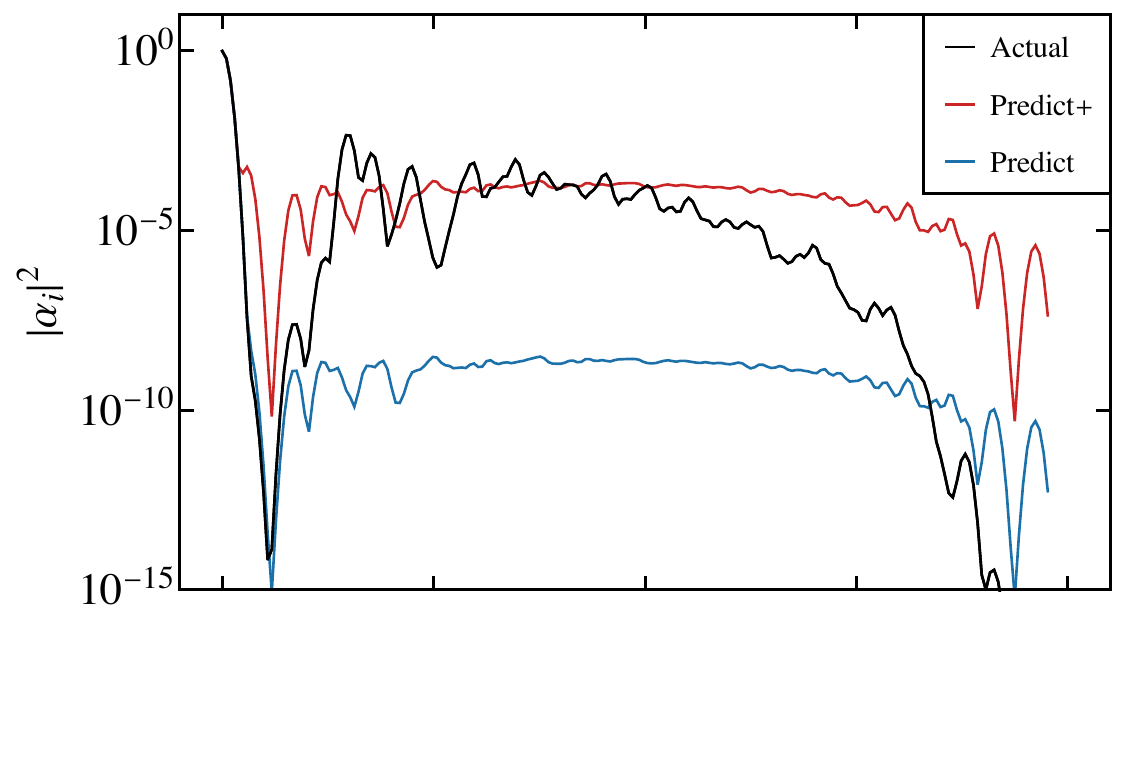}
      \put(3,52.85){\small\bfseries (a)}
    \end{overpic}
    \label{fig:D50}
  \end{subfigure}

  \vspace{-15pt}
    
  \begin{subfigure}[b]{0.475\textwidth}
    \begin{overpic}[width=\linewidth]{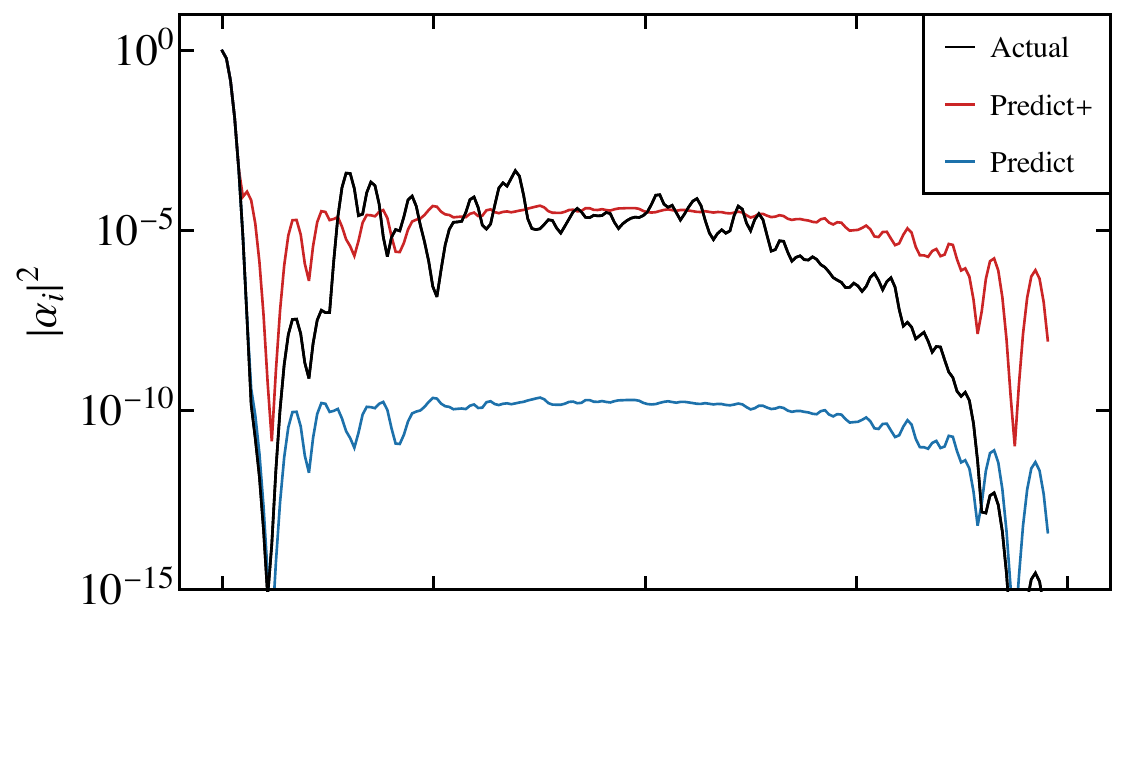}
      \put(3,64){\small\bfseries (b)}
    \end{overpic}
    \label{fig:D100}
  \end{subfigure}
  \vspace{-42pt}
    
  \begin{subfigure}[b]{0.475\textwidth}
    \begin{overpic}[width=\linewidth]{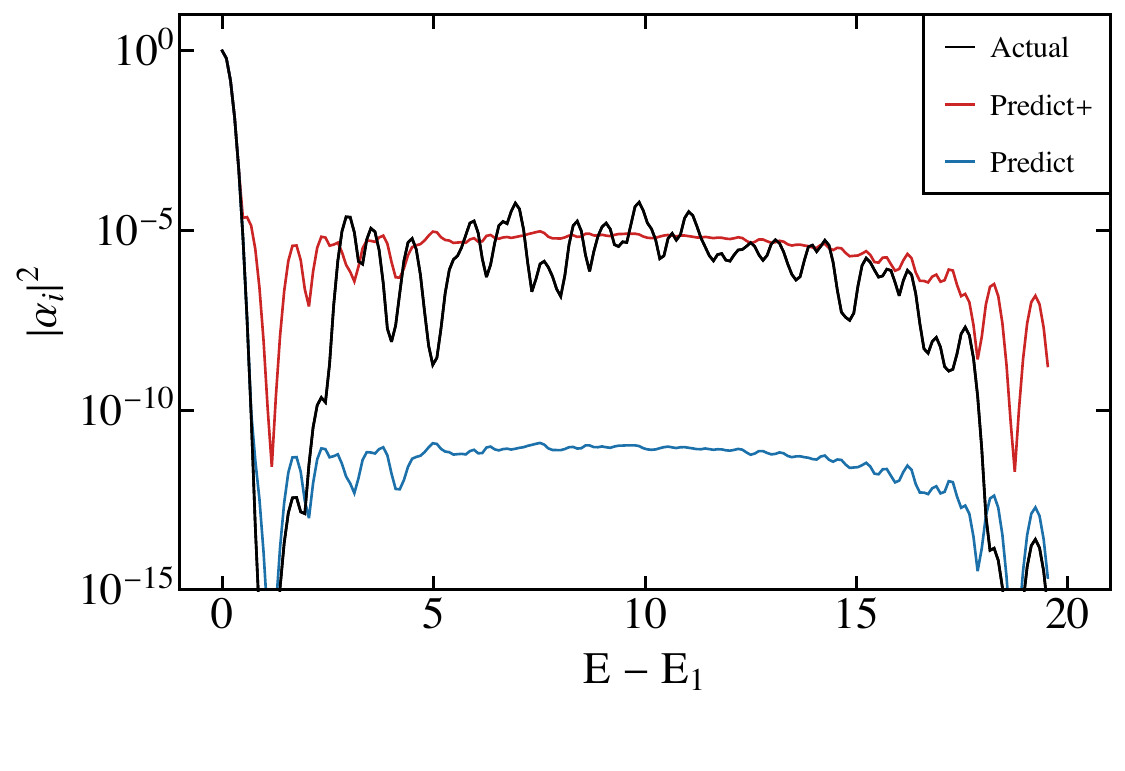}
      \put(3,64){\small\bfseries (c)}
    \end{overpic}
    \label{fig:D150}
  \end{subfigure}
  \vspace{-25pt}
  \captionsetup{justification=raggedright,singlelinecheck=false}
  \caption{%
    Actual versus predicted energy spectra for bond dimension (a) $D=50$, (b) $D=100$, and (c) $D=150$. ``Actual'' denotes the actual spectra of $\ket{\psi(D)}$; ``Predict+'' denotes the predicted spectra when \cref{eq:gs_overlap_bound} is used; ``Predict'' denotes the predicted spectra when \cref{eq:gs_overlap_bound} is unused;  All spectra are broadened by Gaussians with width $0.1$.
  }
  \label{fig:flat_spectrum}
\end{figure}

\cref{fig:flat_spectrum} shows that the predicted spectra exhibit flat tails with heights that decrease with $D$, consistent with the actual spectra \footnote{The predicted and actual spectra in \cref{fig:flat_spectrum} also agree in that they have large dips and peaks at the same locations. This agreement is an artifact of Gaussian broadening and gaps in the eigenvalues of the Hamiltonian.}. As the predicted spectra are based on \cref{prop:energy_min}, the figure supports the claim that the proposition \emph{explains} such flat tails. 

The predicted tail heights agree with the actual tail heights when \cref{eq:gs_overlap_bound} is used, but deviate when it is not. This is because the predicted $p_1\coloneqq \abs{\alpha_1}^2$ values are close to the actual values when \cref{eq:gs_overlap_bound} is used, but deviate when it is not. The relevant values of $p_1$ are given below to $4$~s.f. By definition, $p_1=\Lambda(D)$ when \cref{eq:gs_overlap_bound} is used.

\begin{table}[H]
    \renewcommand{\arraystretch}{1.3} 
    \begin{tabular*}{0.925\linewidth}{@{\extracolsep{\fill}} c c c c @{}}
        \toprule
        & \multicolumn{3}{c}{$p_1$ value} \\ 
        \cmidrule(lr){2-4} 
        $D$ & Actual & \cref{eq:gs_overlap_bound} used ($\Lambda(D)$) & \cref{eq:gs_overlap_bound} unused \\
        \midrule
        $50$  & $0.9879$      & $0.9918$  & $1.000$ \\
        $100$ & $0.9977$      & $0.9983$  & $1.000$  \\
        $150$ & $0.9996$ & $0.9997$ & $1.000$        \\
        \bottomrule
    \end{tabular*}
\end{table}

\paragraph{Energy spectra and quantum advantage.}

Suppose we know the energy spectrum of $\ket{\psi}$ as a function of some compression parameter $m>0$. Then we can calculate the smallest $m$ at which $\ket{\psi}$'s overlap with the ground-energy subspace is ``good'' (call $m_Q$) and at which $\ket{\psi}$'s expected energy is ``good'' (call $m_C$). Since good overlap is necessary for good energy, we have $m_C\geq m_Q$.

To estimate the ground state energy to good accuracy, quantum algorithms can apply phase estimation \cite{qpe_kitaev_1995}, or methods in, e.g., Refs.~\cite{filtering_poulin_2009,gsp_ge_2019,filtering_lin_tong_2020,qet_lin_2022,lowdepth_qpe_wang_2023,filtering_ding_2024}, to any $\ket{\psi}$ that has good overlap.  In contrast, classical algorithms must use a $\ket{\psi}$ that has good energy due to the BQP-completeness of the guided local Hamiltonian problem \cite{dequantize_gharibian_2023,improved_glhp_cade_2023,guide_waite_2025} --- else classical algorithms would be able to efficiently solve any problem in BQP, which is generally believed to be false. 

Therefore, the ratio $m_C/m_Q$ can detect super-polynomial quantum advantage if we assume that the complexity of preparing (a classical description of) $\ket{\psi}$ in both the quantum and classical computational models scales as $\poly(m)$.  This assumption is motivated by its validity when $\ket{\psi}$ is an MPS and $m$ is taken as the maximum bond dimension of $\ket{\psi}$ \cite{mps_schon_2007,state_prep_malz_2024,state_prep_fomichev_2024,state_prep_berry_2025}. Since $m_C$ and $m_Q$ are both calculated from the energy spectrum of a compressed quantum state, this highlights the central role of that spectrum in assessing quantum advantage.

In \cref{app:two_level}, we apply the above reasoning and \cref{prop:energy_min} to two-level Hamiltonians. Such Hamiltonians are analytically tractable and we show that they cannot admit super-polynomial quantum advantage if $m$ can be taken as the stable rank of $\ket{\psi}$.

\paragraph{Conclusion.}

This work gives an analytical framework for predicting the energy spectra of compressed quantum states. The framework explains why such spectra have large support under physically-motivated assumptions. It also explains the main empirical finding of Ref.~\cite{unusual_white_2025} that the spectra of MPSs do not decay exponentially. Further study of energy spectra is well motivated as they enable a principled method for assessing quantum advantage.

\paragraph{Acknowledgments.}

This work is supported by NSERC Grants CRC-2023-00039, RGPIN-2024-06493, and DGECR-2024-00113.

\paragraph{Data availability.}

The data that support the findings of this work are openly available \cite{unusual_spectra_github}.

\vspace{-5pt}
\bibliography{references}

\appendix
\onecolumngrid
\newpage

\section{Volume law for entanglement min-entropy of typical eigenstates}\label{app:volume_law}

\begin{proposition}
Under the setup described in the main text, we have
\begin{equation}
    \Hm(\chi(\ket{\psi_i}),\dots, \chi(\ket{\psi_{2^n}})) \geq 2^{\Omega(n)}.
\end{equation}
\end{proposition}

Recall that the two assumptions we will make are the \emph{eigenstate thermalization hypothesis (ETH)} and a \emph{Gaussian density-of-states}. Under these assumptions, similar results to the proposition are known when the min-entropy (implicit in $\chi$) is replaced by the  von Neumann entropy; see, e.g., the section on entanglement distribution in Ref.~\cite{eigenstate_miao_qiang_2021}, where justification is also given for assuming a Gaussian density-of-states.

\begin{proof}
    Write $N \coloneqq 2^n$.
    Using the definition of the harmonic mean and the fact that $\chi(\ket{\psi_i}) = 2^{S_{\min}(\rho_i)}$, we see that the proposition is equivalent to 
    \begin{equation}
        \frac{2^n}{2^{-S_{\min}(\rho_1)} + \dots + 2^{-S_{\min}(\rho_{N})}} \geq 2^{\Omega(n)}.
    \end{equation}

    Thus it suffices to show there exists constant $\gamma<1$ such that 
    \begin{equation}
        2^{-S_{\min}(\rho_1)} + \dots + 2^{-S_{\min}(\rho_{N})} \leq 2^{\gamma n},
    \end{equation}
    which, under the ETH, is equivalent to 
    \begin{equation}
        2^{-V s_{\min}(\beta_1)} + \dots + 2^{-V s_{\min}(\beta_{N})} \leq 2^{\gamma n},
    \end{equation}
    where $V= c n$ for some constant $c\in (0,1/2]$.
    
    By the definition of $s_{\min}$, showing the last equation is equivalent to showing
    \begin{equation}\label{eq:volume_law_wts}
        2^{-c S_{\min}(\rho(\beta_1))} + \dots + 2^{-c S_{\min}(\rho(\beta_{N}))} \leq 2^{\gamma n},
    \end{equation}
    where $\rho(\beta_i)$ denotes the Gibbs state $\propto \exp(-\beta_i H)$ that has expected energy equal to $E_i$.

    Arrange the $E_i$s so that $E_1 \leq E_2 \leq \dots \leq E_N$. Then for $\beta \in [-\infty, \infty]$,
    \begin{equation}\label{eq:min_entropy_beta}
        S_{\min}(\rho(\beta)) = \frac{1}{\ln 2}\Bigl( \ln\bigl({\textstyle \sum_{i=1}^{N}} \exp(-\beta E_i)\bigr) +  \min_i \beta E_i\Bigr) = \begin{cases}
            \log_2\bigl({\textstyle \sum_{i=1}^{N}} \exp(\beta (E_1-E_i))\bigr) &\text{if $\beta > 0$},
            \\
            n &\text{if $\beta = 0$},
            \\
            \log_2\bigl({\textstyle \sum_{i=1}^{N}} \exp(\beta (E_N-E_i))\bigr) &\text{if $\beta <0$}.
        \end{cases}
    \end{equation}

    \textbf{(}For intuition, the Gibbs state at inverse-temperature $\beta = 0$ describes a uniform distribution over the $\ket{\psi_i}$s, which has expected energy $\sum_i E_i/N$. When the $E_i$s follow a Gaussian distribution, we expect  most $E_i$s to be close to their mean, i.e., $\sum_i E_i/N$. This then means we expect $\beta_i\approx 0$ for most $i$s. Therefore, \cref{eq:min_entropy_beta} implies $S_{\min}(\rho(\beta_i)) \approx n$ for most $i$s. Therefore, the left-hand side of \cref{eq:volume_law_wts} is approximately $N \cdot 2^{-cn} = 2^{(1-c)n}$ which is indeed at most $2^{\gamma n}$ for constant $\gamma \coloneqq (1-c)$. We now proceed to formalize this intuition.\textbf{)}
    
    For each $j\in [N]$, the defining equation for $\beta_j$ is 
    \begin{equation}
        \tr[H\rho(\beta_j)] = E_j, \quad \text{i.e.,} \quad \sum_{i=1}^N E_i \cdot \frac{\exp(-\beta_j E_i)}{\sum_{k=1}^N \exp(-\beta_j E_k)} = E_j,
    \end{equation}
    which can be written as
    \begin{equation}\label{eq:beta}
        \sum_{i=1}^N (E_i - E_j) \exp(-\beta_j E_i) = 0.
    \end{equation}

    To solve the above equation for $\beta_j$, we will use the assumption that the $E_i$s follow a Gaussian. In fact, we can further assume with loss of generality that the $E_i$s follow a \emph{standard} Gaussian with mean $0$ and variance $1$. To see this, observe that if all the $E_i$s are scaled by $a$ and then shifted by $b$ for some real $a,b$ with $a\neq 0$, i.e., $E_i \mapsto E_i' \coloneqq a E_i + b$, the new solution for $\beta_j$ becomes $\beta_j' \coloneqq \beta_j/a$. However, under these transformations, the values of $S_{\min}(\rho(\beta_j))$ will not change, as can be seen directly from  \cref{eq:min_entropy_beta}.

    Then, in the large $N$ limit, we may rewrite \cref{eq:beta} as 
    \begin{equation}
        \frac{N}{\sqrt{2\pi}}\int_{-\infty}^{\infty} (x - E_j) \exp(-\beta_j x) \exp(-x^2/2)\, \mathrm{d}x = 0,
    \end{equation}
    which solves to $\beta_j = - E_j$.
    
    Then, using \cref{eq:min_entropy_beta}, the left-hand-side of \cref{eq:volume_law_wts} can be expressed as
    \begin{equation}
        \LHS \coloneqq \sum_{j\in [N]\colon E_j < 0} \Bigl(\sum_{i=1}^{N} \exp(-E_j (E_1-E_i))\Big)^{-c} + \sum_{j\in [N]\colon E_j \geq 0} \Bigl(\sum_{i=1}^{N} \exp(-E_j (E_N-E_i))\Big)^{-c}.
    \end{equation}
    Since the $E_i$s follow a Gaussian distribution, which is symmetric about $0$, the two summands above must equal and, in the large $N$ limit, we have
    \begin{equation}
         \LHS  = 2 N \int_{x = 0}^{\infty} \Bigl(N \int_{y=-\infty}^{\infty} \exp(-x (E_N - y)) \cdot p(y) \, \mathrm{d}y \Bigr)^{-c} \, p(x) \, \mathrm{d}x,
    \end{equation}
    where $p(\cdot)$ is the probability density function of the standard Gaussian.

    By direct calculation (using Mathematica, say), we have
    \begin{align}
        \LHS=&~2 N^{1-c} \int_{x = 0}^{\infty} \Bigl(\exp(-x E_N +x^2/2) \Bigr)^{-c} \, p(x) \, \mathrm{d}x
        \notag
        \\
        =&~2 N^{1-c} \frac{1}{\sqrt{2\pi}}\int_{x = 0}^{\infty} \exp(cx E_N -cx^2/2) \, \exp(-x^2/2) \, \mathrm{d}x
        \notag
        \\
        =&~\frac{2N^{1-c}}{\sqrt{c+1}} \exp\Bigl(\frac{c^{2}E_N^{2}}{2(c+1)}\Bigr) \Phi\Bigl(\frac{cE_N}{\sqrt{c+1}}\Bigr), \label{eq:LHS_simplified}
    \end{align}
    where $\Phi$ is the cumulative density function of the standard Gaussian.

    Now, since the $E_i$s follow a standard Gaussian, their values can be determined by the condition that the $N+1$ intervals 
    \begin{equation}
        (-\infty, E_1], \, (E_1, E_2], \dots, (E_{N-1},E_N], \, (E_N, \infty)
    \end{equation}
    should all carry equal probability mass, $1/(N+1)$, under the standard Gaussian distribution.

    Let $Z$ denote a standard Gaussian random variable. Then, for $z>0$, Chernoff's bound gives
    \begin{equation}
        \Pr[Z \geq z] \leq \exp(-z^2/2).
    \end{equation}
    Choosing $z = \sqrt{2\cdot \ln(2N)}$, we deduce
    \begin{equation}
         \Pr[Z \geq \sqrt{2\cdot \ln(2N)}] \leq \frac{1}{2N} < \frac{1}{N+1},
    \end{equation}
    and therefore 
    \begin{equation}\label{eq:EN_bound}
        E_N \leq  \sqrt{2\cdot \ln(2N)}.
    \end{equation}

    Substituting \cref{eq:EN_bound} into \cref{eq:LHS_simplified} and using $\Phi(\cdot) \leq 1$, we obtain constants $c'>0$ and $\gamma <1$ such that 
    \begin{equation}\label{eq:volume_law_wts_done}
        \LHS \leq c' \cdot N^{1-c+c^2/(c+1)} \leq c' \cdot N^{1-c/2} = c'\cdot  2^{(1-c/2)n} \leq 2^{\gamma n}, 
    \end{equation}
    where the second inequality uses $c^2/(c+1) \leq c/2$ for $c\in (0,1/2]$. 

    \cref{eq:volume_law_wts_done} completes the proof because it is equivalent to the desired \cref{eq:volume_law_wts}.
\end{proof}

\section{\texorpdfstring{Proof of Proposition 2}{}}
\label{app:proof_energymin}

\energymin*

The following proof uses duality theory for convex optimization; see Ref.~\cite[Chapter 5]{cvx_boyd_vandenberghe_2004} for all necessary background.

\begin{proof}[Proof of \cref{prop:energy_min}.] We write $p_i\coloneqq \abs{\alpha_i}^2$. Then the minimization problem is equivalent to 
    \begin{alignat}{2}
        &\mathrm{minimize} \quad &&\sum_{i=1}^k p_i E_i\label{eq:primal_copy}
        \\
        &\mathrm{s.t.} \quad && \sum_{i=1}^k\sqrt{\frac{p_i}{M_i}} \, \geq \, \frac{1}{\sqrt{m}}\notag
        \\
        &\quad && \forall i\in [k], \ p_i\geq 0\notag
        \\
        &\quad && \sum_{i=1}^k p_i = 1.\notag
    \end{alignat}
    
    This is a convex optimization problem with domain $\mathbb{R}_{\geq 0}^k$. Observe that the nontriviality condition on $k$, \cref{eq:non_triviality_k}, implies that $p \coloneqq Z^{-1}\cdot(1/M_1,\dots, 1/M_k)\in \mathbb{R}_{>0}^k$, where $Z\coloneqq \sum_{i=1}^k 1/M_i$, is a strictly feasible solution. Therefore, Slater's condition is satisfied, which implies that the dual optimum equals the primal optimum given by \cref{eq:primal_copy}. We proceed to show that the dual optimum can be written as \cref{eq:dual}.
    
    The associated Lagrangian $L\colon \mathbb{R}^k_{\geq0} \times \mathbb{R}^{k+1} \times \mathbb{R} \to \mathbb{R}$ is defined by 
    \begin{align}
        L(p,\lambda,\nu) &\coloneqq \sum_{i=1}^k p_i E_i - \lambda_0\Bigl(
        \sum_{i=1}^k \sqrt{\frac{p_i}{M_i}} - \frac{1}{\sqrt{m}} \Bigr) - \sum_{i=1}^k \lambda_i p_i -\nu\Bigl(\sum_{i=1}^k p_i -1 \Bigr).
    \end{align}
    
    The dual function $g\colon \mathbb{R}^{k+1}\times \mathbb{R} \to \mathbb{R}\cup \{-\infty\}$ is defined by 
    \begin{equation}
        g(\lambda,\nu) \coloneqq \min_{p \in \mathbb{R}^k_{\geq0}} L(p,\lambda,\nu).
    \end{equation}
    To compute $g$, we first rewrite
    \begin{align}
        g(\lambda,\nu) =&~\nu + \frac{\lambda_0}{\sqrt{m}} + \sum_{i=1}^k \min_{p_i\geq 0}\Big(p_i(E_i-\lambda_i-\nu) - \lambda_0\sqrt{\frac{p_i}{M_i}}\Bigr).\label{min_pi}
    \end{align}
    We analyze according to the following cases:
    \begin{enumerate}[leftmargin=*]
        \item Case $E_i-\lambda_i-\nu<0$ for some $i\in[k]$: $g(\lambda,\nu) = -\infty$. 
        \item Case $E_i-\lambda_i-\nu\geq 0$ for all $i\in[k]$:
        \begin{enumerate}[leftmargin=*]
            \item If $\lambda_0\leq 0$, then $g(\lambda,\nu) = \nu + \lambda_0/\sqrt{m}$. 
            \item If $\lambda_0>0$ and $E_i-\lambda_i-\nu = 0$, then $g(\lambda,\nu) = -\infty$.
            \item If $\lambda_0>0$ and $E_i-\lambda_i-\nu > 0$, then we can explicitly solve the minimization problem in \cref{min_pi} for each $p_i$ independently to obtain
            \begin{equation}\label{eq:pi_form}
                p_i = \frac{\lambda_0^2}{4M_i(E_i-\lambda_i-\nu)^2} \quad \text{ for all $i\in[k]$},
            \end{equation}
            which yields
            \begin{equation}
                g(\lambda,\nu) = \nu + \frac{\lambda_0}{\sqrt{m}} - \sum_{i=1}^k \frac{\lambda_0^2}{4M_i(E_i-\lambda_i-\nu)}.
            \end{equation}
        \end{enumerate}
    \end{enumerate}

    Therefore, the dual problem, i.e.,
    \begin{alignat}{2}
        &\mathrm{maximize} \quad &&g(\lambda,\nu)
        \\
        &\mathrm{s.t.} \quad && \lambda \in \mathbb{R}^{k+1}_{\geq 0},\notag
    \end{alignat}
    can be expressed as
    \begin{alignat}{2}
        &\mathrm{maximize} \quad &&\nu + \frac{\lambda_0}{\sqrt{m}} - \sum_{i=1}^k \frac{\lambda_0^2}{4M_i(E_i-\lambda_i-\nu)}\label{eq:dual_with_lambdai}
        \\
        &\mathrm{s.t.} \quad && \lambda \in \mathbb{R}^{k+1}_{\geq 0}\notag
        \\
        & \quad &&\forall i\in [k], \, E_i-\lambda_i - \nu > 0;\notag
    \end{alignat}
    which, using $E_1 = \min_{i=1}^k E_i$, is equivalent to 
    \begin{alignat}{2}
        &\mathrm{maximize} \quad &&\nu + \frac{\lambda_0}{\sqrt{m}} - \sum_{i=1}^k \frac{\lambda_0^2}{4M_i(E_i-\nu)}\label{eq:dual_without_lambdai}
        \\
        &\mathrm{s.t.} \quad && \lambda_0 \in \mathbb{R}_{\geq 0}\notag
        \\
        & \quad &&\nu < E_1;\notag
    \end{alignat}
    which, upon maximizing the quadratic in $\lambda_0$ for fixed $\nu$, is equivalent to
    \begin{equation}\label{eq:dual_copy}
        \max_{\nu < E_1} \quad  \nu +\Bigl(m \sum_{i=1}^k \frac{1}{M_i(E_i - \nu)}\Bigr)^{-1}.
    \end{equation}

    To see that the maximization in \cref{eq:dual_copy} is uniquely attained at the unique solution $\nu^*<E_1$ satisfying \cref{eq:dual_optimum_condition}, it suffices to show that the function $h\colon (-\infty,E_1)\to \mathbb{R}$ defined by 
   \begin{equation}
       h(\nu) \coloneqq \nu +\Bigl(m \sum_{i=1}^k \frac{1}{M_i(E_i - \nu)}\Bigr)^{-1}
   \end{equation}
    satisfies
    \begin{enumerate}
        \item $h'(\nu^*) = 0$ is equivalent to \cref{eq:dual_optimum_condition};
        \item $h$ is strictly concave;
        \item $\lim_{\nu \rightarrow E_1^-}h'(\nu) < 0$; and
        \item $\lim_{\nu \rightarrow -\infty}h'(\nu) > 0$.
    \end{enumerate}
    
    We prove each item in turn. For convenience, for $p\in \mathbb{N}$ and $\nu<E_1$, we write
   \begin{equation}
       S_p(\nu) \coloneqq \sum_{i=1}^k \frac{1}{M_i(E_i-\nu)^p}.
   \end{equation}

    \begin{enumerate}
    \item To see $h'(\nu^*) = 0$ is equivalent to \cref{eq:dual_optimum_condition}, observe that
   \begin{equation}
       h'(\nu) = 1 - \frac{1}{m}\frac{S_2(\nu)}{S_1(\nu)^2},
   \end{equation} 
   so $h'(\nu^*) = 0$ is clearly equivalent to \cref{eq:dual_optimum_condition}.
    
    \item To see $h$ is strictly concave, observe that
   \begin{equation}
       h''(\nu) = -\frac{2}{m}\frac{\, S_1(\nu) \, S_3(\nu) - S_2(\nu)^2}{S_1(\nu)^3}.
   \end{equation}
   Then observe that for each fixed $\nu<E_1$, $S_2(\nu)^2 \leq S_1(\nu)S_3(\nu)$ by the Cauchy-Schwarz inequality and that equality holds if and only if $E_i = E_1$ for all $i\in [k]$, that is, $l=k$. But $l=k$ is forbidden by combining \cref{eq:non_triviality_l,eq:non_triviality_k}, so we must have $S_2(\nu)^2 < S_1(\nu)S_3(\nu)$.
   Therefore, $h''(\nu)<0$ for all $\nu<E_1$, so $h$ is strictly concave. 

    \item To see $\lim_{\nu \rightarrow E_1^-}h'(\nu) < 0$, observe that
   \begin{align}
       \lim_{\nu \rightarrow E_1^-}h'(\nu) =&~1 - \frac{1}{m\cdot \sum_{i=1}^l M_i^{-1}} = 1 - \frac{\Hm(M_1,\dots,M_l)}{m\cdot  l}<0,
   \end{align}
   where the last inequality uses the nontriviality condition on $l$, \cref{eq:non_triviality_l}.

    \item To see $\lim_{\nu \rightarrow -\infty} h'(\nu) >0$, observe that 
   \begin{equation}
       h'(-\nu) = 1 - \frac{1}{m}\frac{\sum_{i=1}^k M_i^{-1}(1+E_i/\nu)^{-2}}{\bigl(\sum_{i=1}^k M_i^{-1}(1+E_i/\nu)^{-1}\big)^2},
   \end{equation}
   and so 
   \begin{align}
      \lim_{\nu \rightarrow -\infty}h'(\nu) =&~ 
      \lim_{\nu \rightarrow \infty}h'(-\nu) = 1 - \frac{1}{m}\frac{1}{\sum_{i=1}^k M_i^{-1}} = 1 -\frac{1}{k} \frac{\Hm(M_1,\dots,M_k)}{m} > 0,
   \end{align}
   where the last inequality uses the nontriviality condition on $k$, \cref{eq:non_triviality_k}.
    \end{enumerate}

   The ``moreover'' part of the proposition follows from \cref{eq:pi_form} and the fact that the dual optimizer has $\lambda_i =0$ for all $i>0$, as can be seen from the equivalence between \cref{eq:dual_with_lambdai} and \cref{eq:dual_without_lambdai}.
\end{proof}

\section{Two-level Hamiltonians}\label{app:two_level}

We illustrate how a characterization of the energy spectra of a compressed quantum state $\ket{\psi}$ can be used to assess the quantum advantage for ground state energy estimation via $\ket{\psi}$. As a concrete example, we consider Hamiltonians with exactly two distinct energy levels, in which case we can \emph{exactly solve} \cref{eq:dual_optimum_condition} of \cref{prop:energy_min} to obtain the predicted energy spectra.

Let $H$ be any $n$-qubit Hamiltonian having exactly two distinct energy levels $\xi_1$ and $\xi_2$, with $\Delta \coloneqq \xi_2 - \xi_1 > 0$.  Write $k\coloneqq 2^n$. For $i\in [k]$, write $\ket{\psi_i}$ for the $i$th energy eigenstate of $H$, and write $E_i$ for its energy and $M_i$ for $\chi(\ket{\psi_i})$. Write $\ket{\psi} \coloneqq \sum_{i=1}^k \alpha_i\ket{\psi_i}$, where the $\alpha_i$s are a minimizer of \cref{prog:primal}, and write $m$ for $\chi(\ket{\psi})$. 

Write
\begin{equation}
    a_1~\coloneqq~\sum_{i\in [k]\colon E_i = \xi_1} \frac{1}{M_i}, \quad a_2~\coloneqq~\sum_{i\in [k]\colon E_i = \xi_2} \frac{1}{M_i}.
\end{equation}

Write 
\begin{align}
    E~\coloneqq&~\bra{\psi}H\ket{\psi},
    \\[1.5ex]
    p~\coloneqq&~ \sum_{i\in [k]\colon E_i = \xi_1}\abs{\braket{\psi|\psi_i}}^2,
\end{align}
so that $E$ is the expected energy of $\ket{\psi}$ and $p$ is the overlap of $\ket{\psi}$ with the ground state energy subspace.

We assume the nontriviality conditions, \cref{eq:non_triviality_l,eq:non_triviality_k}, hold. In this case, it will be convenient to parametrize $m$ using a real parameter $\mu \in (0,1)$ as 
$m = (a_1+\mu a_2)/(a_1(a_1+a_2))$,
where the condition $\mu\in (0,1)$ is equivalent to the nontriviality conditions.

Then, solving \cref{eq:dual_optimum_condition} of \cref{prop:energy_min}, we derive
\begin{align}
    p~=&~\frac{(a_1+\sqrt{\mu} \,  a_2)^2}{(a_1+a_2)(a_1+\mu a_2)},\label{eq:pmin}
    \\[1ex]
    E~=&~\xi_1 + \Delta \cdot (1-p).\label{eq:Emin}
\end{align}

\cref{eq:pmin,eq:Emin} can be used to assess quantum advantage for the problem $\calP$ of estimating the ground state energy of $H$, i.e., $\xi_1$, to \emph{constant additive precision} in the case $\norm{H}$ and $\Delta$ are both extensive with $n$, i.e., $\norm{H},\Delta=\poly(n)$. The notation $\poly(n)$ refers to some polynomial in $n$ and different occurrences of $\poly(n)$ may refer to different polynomials. $\calP$ is a quantitative version of the ground state energy estimation problem that is of interest in quantum physics and chemistry~\cite{dequantize_gharibian_2023,improved_glhp_cade_2023,guide_waite_2025}. (Refs.~\cite{dequantize_gharibian_2023,improved_glhp_cade_2023,guide_waite_2025} follow the convention of using a normalized $H$, i.e., $\norm{H} = 1$, so the problem they consider of ground state energy estimation to $1/poly(n)$ additive precision translates to our $\calP$, which concerns \emph{constant} additive precision, under our convention of $\norm{H} = \poly(n)$.) 

Quantum algorithms can always solve $\calP$ by preparing a $\ket{\psi}$ with $p \geq 1/\poly(n)$ and applying phase estimation \cite{qpe_kitaev_1995}, or methods in, e.g., Refs.~\cite{filtering_poulin_2009,gsp_ge_2019,filtering_lin_tong_2020,qet_lin_2022,lowdepth_qpe_wang_2023,filtering_ding_2024}. On the other hand, classical algorithms can always solve $\calP$ by preparing a (classical description of) $\ket{\psi}$ with $E \leq \xi_1 + O(1)$, which translates to $p \geq 1-O(1/\Delta) = 1-1/\poly(n)$ via \cref{eq:Emin}, and applying expected energy estimation. (Refs.~\cite{dequantize_gharibian_2023,improved_glhp_cade_2023,guide_waite_2025} provide evidence that classical algorithms \emph{cannot} always solve $\calP$ using a $\ket{\psi}$ with $p \geq 1/\poly(n)$ by showing that this problem is BQP-complete.) 

The dominant cost in solving $\calP$ is that of preparing $\ket{\psi}$ in both the quantum and classical cases. We assume that the cost of preparing $\ket{\psi}$ scales as $\poly(mn)$. This assumption is motivated by its validity when $\ket{\psi}$ is an MPS and $m$ is replaced by the maximum bond dimension of $\ket{\psi}$ \cite{mps_schon_2007,state_prep_malz_2024,state_prep_fomichev_2024,state_prep_berry_2025}. Then, quantum advantage can be assessed by the ratio $\poly(m_C/m_Q)$ up to a multiplicative factor of $\poly(n)$, where $m_Q$ is the smallest value at which $p \geq 1/\poly(n)$ and $m_C$ is the smallest value at which $p \geq  1-1/\poly(n)$.

Our primary objective is to assess whether super-polynomial quantum advantage exists, so it suffices to define quantum advantage by $m_C/m_Q$. The values of $m_Q$ and $m_C$ can be calculated by inverting \cref{eq:pmin} to express $m$ in terms of $p \in (a_1/(a_1+a_2),1)$ as:
\begin{equation}\label{eq:m}
m =
1 \Big/\bigl( \, \sqrt{\vphantom{a_2(1-p)} a_1 p} + \sqrt{a_2 (1 - p)} \, \bigr)^2.
\end{equation}

Unfortunately, the functional form of \cref{eq:m} precludes super-polynomial quantum advantage as can be seen by the following case analysis.
\begin{enumerate}[leftmargin=*]
    \item Case $a_1 \geq a_2$. In this case,
    \begin{equation*}
        a_1^{-1}(\sqrt{\vphantom{1-p}p}+\sqrt{1-p})^{-2}\leq m \leq a_1^{-1} p^{-1}.
    \end{equation*}
    Therefore,
    \begin{align*}
        \frac{m_C}{m_Q} \leq\frac{(1-1/\poly(n))^{-1}}{(\sqrt{1/\poly(n)} + \sqrt{1- 1/\poly(n)})^{-2}}\leq O(1).
    \end{align*}\vspace{1ex}
    \item Case $a_1 \leq a_2$. In this case,
    \begin{equation*}
        a_2^{-1}(\sqrt{\vphantom{1-p}p}+\sqrt{1-p})^{-2}\leq m \leq a_2^{-1} (1-p)^{-1}.
    \end{equation*}
    Therefore,
    \begin{align*}
        \frac{m_C}{m_Q} \leq\frac{(1/\poly(n))^{-1}}{(\sqrt{1/\poly(n)} + \sqrt{1- 1/\poly(n)})^{-2}}\leq \poly(n).
    \end{align*}
\end{enumerate}

However, this analysis is for two-level Hamiltonians only. It would be interesting for future work to discover more sophisticated Hamiltonians with $M_i$ and $E_i$ profiles that could lead to super-polynomial quantum advantage.

\end{document}